\documentclass[a4paper]{article}

\usepackage{ae,aecompl}
\usepackage{amsthm,amssymb,amsmath,amsfonts}
\usepackage{multirow,booktabs}
\usepackage[ruled]{algorithm}
\usepackage{algorithmicx}
\usepackage{algpseudocode}
\usepackage{color}

\theoremstyle{plain}
\newtheorem{theorem}{Theorem}
\newtheorem{fact}[theorem]{Fact}

\theoremstyle{definition}

\newcommand{\A}{A}

\newcommand{\ADV}{\textsc{Adv}}
\newcommand{\OPT}{\textsc{Opt}}
\newcommand{\ALG}{\textsc{Alg}}
\newcommand{\oldMIX}{\textsc{RMix}}
\newcommand{\newMIX}{\textsc{Mix-R}}
\newcommand{\GALG}{\mathcal{G}_\mathrm{\ALG}}
\newcommand{\GGADV}{\mathcal{G}_\mathrm{\ADV}}
\newcommand{\GADV}[1]{\mathcal{G}_\mathrm{\ADV}^{({#1})}}
\newcommand{\GRMIX}{\mathcal{G}_{\newMIX}}
\newcommand{\GOPT}{\mathcal{G}_\mathrm{\OPT}}
\newcommand{\R}{\mathcal{R}}
\newcommand{\E}{\mathbb{E}}

\newcommand{\braced}[1]{{\left\{#1\right\}}}
\newcommand{\emdash}{\hspace{1pt}---\hspace{1pt}}

\newcommand{\mytable}[3]{
\begin{table}[t]
\centering
#3
\caption{#2}
\label{tbl:#1}
\end{table}
}

\newcommand{\kw}[1]{{\textbf{#1}}}

\title{One to Rule Them All:\\
	A~General Randomized Algorithm for Buffer Management with Bounded Delay}

\author{\L{}ukasz Je\.z\thanks{University of Wroc\l{}aw.
	Work supported by MNiSW grant N~N206~490638, 2010--2011.}}

\begin{document}

\maketitle

\begin{abstract}
We give a~memoryless scale-invariant randomized algorithm {\newMIX} for
\emph{buffer management with bounded delay} that is $e/(e-1)$-competitive
against an adaptive adversary, together with better performance guarantees
for many restricted variants, including the $s$-bounded instances.
In particular, {\newMIX} attains the optimum competitive ratio of $4/3$ on
$2$-bounded instances.

Both {\newMIX} and its analysis are applicable to a~more general problem,
called \emph{Collecting Items}, in which only the relative order between packets'
deadlines is known. {\newMIX} is the optimal \emph{memoryless} randomized
algorithm against adaptive adversary for that problem in a~strong sense.

While some of provided upper bounds were already known, in general,
they were attained by several different algorithms.
\end{abstract}

\section{Introduction}

In this paper, we consider the problem of \emph{buffer management with bounded
delay}, introduced by Kesselman et al.~\cite{diffserv-kesselman}.  This problem
models the behavior of a~single network switch responsible for scheduling packet
transmissions along an outgoing link as follows. We assume that time is
divided into unit-length steps. At the beginning of a~time step, any number of
packets may arrive at a~switch and be stored in its \emph{buffer}.  Each packet
has a~positive weight, corresponding to the packets priority, and a~deadline,
which specifies the latest time when the packet can be transmitted. Only one packet
from the buffer can be transmitted in a~single step. A~packet is removed from the
buffer upon transmission or expiration, i.e., reaching its deadline.
The goal is to maximize the \emph{gain}, defined as the total weight of the packets
transmitted.

We note that \emph{buffer management with bounded delay} is equivalent to
a~scheduling problem in which packets are represented as
jobs of unit length, with given
release times, deadlines and weights; release times and deadlines are restricted
to integer values.  In this setting, the goal is to maximize the total weight
of jobs completed before their respective deadlines.

As the process of managing packet queue is inherently a~real-time task, we
model it as an \emph{online problem}.  This means that the algorithm, when deciding
which packets to transmit, has to base its decision solely on the packets which
have already arrived at a~switch, without the knowledge of the future.



\subsection{Competitive Analysis}\label{sec: comp-anal}
To measure the performance of an online algorithm, we use the standard notion
of \emph{competitive analysis}~\cite{borodin-book}, which, roughly speaking, compares
the gain of the algorithm to the gain of the \emph{optimal solution} on the same instance.
For any algorithm $\ALG$, we denote its gain on instance $I$ by
$\GALG(I)$. The optimal offline algorithm is denoted by $\OPT$. We say that
a~deterministic algorithm $\ALG$ is $\R$-competitive if on any
instance~$I$ it holds that $\GALG(I) \geq \frac 1 \R \cdot \GOPT(I)$. 

When analyzing the performance of an online algorithm $\ALG$, we view the
process as a~game between $\ALG$ and an {\em adversary}. The adversary controls
what packets are injected into the buffer and chooses which of them to send.
The goal is then to show that the adversary's gain is at most $\R$ times $\ALG$'s gain.

If the algorithm is randomized, we consider its expected gain, $\E[\GALG(I)]$,
where the expectation is taken over all possible random choices made by $\ALG$.
However, in the randomized case, the power of the adversary has to be further
specified.  Following Ben-David et al.~\cite{relating-adversaries}, we
distinguish between an \emph{oblivious} and an \emph{adaptive-online} adversary,
which from now on we will call \emph{adaptive}, for short.
An oblivious adversary has to construct the whole
instance in advance. This instance may depend on $\ALG$ but not on
the random bits used by $\ALG$ during the computation. 
The expected gain of $\ALG$ is compared to the gain of the optimal
offline solution on $I$.  In contrast, in case of an adaptive adversary, the
choice of packets to be injected into the buffer may depend on the algorithm's 
behavior up to the given time step. This adversary must also provide
an answering entity~$\ADV$, which
creates a~solution in parallel to $\ALG$. This solution may not be changed
afterwards.  We say that $\ALG$ is $\R$-competitive against an adaptive
adversary if for any adaptively created instance $I$ and any answering
algorithm $\ADV$, it holds that $\E[\GALG(I)] \geq \frac{1}{\R} \cdot
\E[\GGADV(I)]$. We note that $\ADV$ is (wlog) deterministic,
but as $\ALG$ is randomized, so is the instance~$I$.

In the literature on online algorithms~\cite{borodin-book}, the
definition of the competitive ratio sometimes allows an additive constant,
i.e.,~a~deterministic algorithm is then called $\R$-competitive if there exists
a~constant $\alpha~\geq 0$ such that for any instance~$I$ it holds that
$\GALG(I) \geq \frac 1 \R \cdot \GOPT(I) - \alpha$.  An~analogous definition
applies to the randomized case.  For our algorithm {\newMIX} the bound holds for
$\alpha~= 0$, which is the best possible.


\subsection{Basic Definitions}\label{sec: prel}

We denote a~packet with \emph{weight} $w$ and \emph{relative deadline} $d$ by
$(w,d)$, where the relative deadline of a~packet is, at any time, the number
of steps after which it expires. The packet's \emph{absolute deadline},
on the other hand, is the exact point in time at which the packet expires.
a~packet that is in the buffer, i.e., has already been released and has neither
expired nor been transmitted by an algorithm, is called \emph{pending} for the
algorithm. The \emph{lifespan} of a~packet is its relative deadline value upon
injection, or in other words the difference between its absolute deadline and
release time.

The goal is to maximize the weighted throughput, i.e., the total weight of
transmitted packets.
We assume that time is slotted in the following way. We distinguish between
points in time and time intervals, called \emph{steps}. In step $t$,
corresponding to the interval $(t,t+1)$, $\ADV$ and the algorithm choose,
independently, a~packet from their buffers and transmit it. The packet transmitted
by the algorithm ({\ADV}) is immediately removed from the buffer and no longer
pending. Afterwards, at time $t+1$, the relative deadlines of all remaining
packets are decremented by $1$, and the packets whose relative deadlines reach
$0$ expire and are removed from both {\ADV}'s and the algorithm's buffers.
Next, the adversary injects any set of packets. At this point, we proceed to step
$t+1$.

To no surprise, all known algorithms are \emph{scale-invariant}, which means
that they make the same decisions if all the weights of packets in an 
instance are scaled by a~positive constant.  a~class of further restricted
algorithms is of special interest for their simplicity.  An algorithm is
\emph{memoryless} if in every step its decision depends only on the set
of packets pending at that step. An algorithm that is both memoryless and
scale-invariant is called \emph{memoryless scale-invariant}.


\subsection{Previous and Related Work, Restricted Variants}

The currently best, $1.828$-competitive, deterministic algorithm for general
instances was given by Englert and Westermann~\cite{bounded-delay-westermann}.
Their algorithm is scale-invariant, but it is \emph{not} memoryless. However,
in the same article Englert and Westermann provide another, $1.893$-competitive,
deterministic algorithm that is memoryless scale-invariant.
The best known randomized algorithm is the $1.582$-competitive memoryless
scale-invariant {\oldMIX}, proposed by~Chin~et~al.~\cite{bounded-delay-chin-journal}.
For reasons explained in Section~\ref{sec: analysis-intro} the original analysis
by Chin~et~al. is only applicable in the oblivious adversary model.
However, a~refined analysis shows that the algorithm remains $1.582$-competitive
in the adaptive adversary model~\cite{rmix-arxiv}.

Consider a~(memoryless scale-invariant) greedy algorithm that always transmits
the heaviest pending packet. It is not hard to observe that it is $2$-competitive,
and actually no better than that. But for a~few years no better deterministic
algorithm for the general case was known. This naturally led to a~study of many
restricted variants. Below we present some of them, together with known results.
The most relevant bounds known are summarized in Table~\ref{tbl:results}. Note
that the majority of algorithms are memoryless scale-invariant.

For a~general overview of techniques and results on buffer management,
see the surveys by Azar~\cite{packet-scheduling-azar-survey},
Epstein and Van Stee~\cite{EpsSte04A}
and Goldwasser~\cite{packet-scheduling-goldwasser-survey}.

\paragraph{Uniform Sequences}

An instance is \emph{$s$-uniform} if the lifespan of each packet is exactly~$s$.
Such instances have been considered for two reasons. Firstly, there is a~certain
connection between them and the \emph{FIFO model} of buffer management, also
considered by Kesselmann~et~al.~\cite{diffserv-kesselman}. Secondly, the
$2$-uniform instances are among the most elementary restrictions that do not
make the problem trivial. However, analyzing these sequences is not easy:
while a~simple deterministic $1.414$-competitive algorithm for $2$-uniform
instances~\cite{bounded-delay-anzhu} is known to be optimal among memoryless
scale-invariant algorithms~\cite{bounded-delay-chin-journal}, for unrestricted
algorithms a~sophisticated analysis shows the optimum competitive ratio is
$1.377$~\cite{bounded-delay-chrobak-improved}.

\paragraph{Bounded Sequences}

An instance is \emph{$s$-bounded} if the lifespan of each packet is at most~$s$;
therefore every $s$-uniform instances is also $s$-bounded. This class of instances
is important, because the strongest lower bounds on the competitive ratio known
for the problem employ $2$-bounded instances. These are $\phi\approx 1.618$ for
deterministic algorithms~\cite{diffserv-np-single,chin-timesharing,hajek-scheduling},
$1.25$ for randomized algorithms in the oblivious adversary model~\cite{chin-timesharing},
and $4/3$ in the adaptive adversary model~\cite{bounded-delay-2-bounded}.
For $2$-bounded instances algorithms matching these bounds are
known~\cite{diffserv-kesselman,bounded-delay-chin-journal,bounded-delay-2-bounded}.
A~$\phi$-competitive deterministic algorithm is also known for $3$-bounded
instances~\cite{bounded-delay-chin-journal}, but in general the best
algorithms for $s$-bounded instances are only known to be
$2-2/s+o(1/s)$-competitive~\cite{bounded-delay-chin-journal}.

\paragraph{Similarly Ordered Sequences}

An instance is \emph{similarly ordered} or has \emph{agreeable deadlines} if
for every two packets $i$ and $j$ their spanning intervals are not properly
contained in one another, i.e., if $r_i < r_j$ implies $d_i \leq d_j$. Note
that every $2$-bounded instance is similarly ordered, as is every $s$-uniform
instance, for any $s$. An optimal deterministic $\phi$-competitive
algorithm~\cite{bounded-delay-li} and a~randomized $4/3$-competitive algorithm
for the oblivious adversary model~\cite{agreeable-moje} are known.
With the exception of $3$-bounded instances, this is the most general class
of instances for which a~$\phi$-competitive deterministic algorithm is known.

\paragraph{Other restrictions}

Among other possible restrictions, let us mention one for which our algorithm
provides some bounds. Motivated by certain transmission protocols, which usually
specify only several priorities for packets, one might bound the number of
different packet weights. In fact, Kesselmann~et~al. considered deterministic
algorithms for instances with only two distinct packet
weights~\cite{diffserv-kesselman}.

\paragraph{Generalization: \emph{Collecting Weighted Items from a~Dynamic Queue}}
\label{sec: items}

Bienkowski~et~al.~\cite{item-collecting} studied
a~generalization of buffer management with bounded delay, in which the algorithm
knows only the relative order between packets' deadlines rather than their exact
values; after Bienkowski~et~al. we dub the generalized problem \emph{Collecting
Items}. Their paper focuses on deterministic algorithms but it does provide
certain lower bounds for memoryless algorithms, matched by our algorithm.
See Appendix~\ref{sec:lb-in-appx} for details.


\mytable{results}{Comparison of known and new results.
The results of this paper are shown in boldface; a~reference next to such entry
means that this particular bound was already known. The results without
citations are implied by other entries of the table. An asterisk denotes that
the algorithm attaining the bound is memoryless scale-invariant.}
{
\begin{tabular}{lllll}
\toprule
 && deterministic & (rand.) adaptive & (rand.) oblivious \\
\midrule 
\multirow{2}{*}
{general}
 	& upper & $1.828$~\cite{bounded-delay-westermann}, $1.893^*$~\cite{bounded-delay-westermann} & {\boldmath $1.582^*$~\cite{rmix-arxiv}} & $1.582^*$~\cite{bounded-delay-chin-journal}\\
 	&lower & 1.618 & 1.333 & 1.25 \\
\midrule 
\multirow{2}{*}
{$s$-bounded}
 	& upper & $2-\frac{2}{s}+o(\frac{1}{s})^*$~\cite{bounded-delay-chin-journal} & {\boldmath $1/\left(1-(1-\frac{1}{s})^s\right)^*$} & $1/\left(1-(1-\frac{1}{s})^s\right)^*$ \\
	& lower & 1.618 & 1.333 & 1.25 \\
\midrule 
\multirow{2}{*}
{2-bounded}
 	& upper & $1.618^*$~\cite{diffserv-kesselman} &  {\boldmath $1.333^*$~\cite{bounded-delay-2-bounded}} & $1.25^*$~\cite{bounded-delay-chin-journal} \\
	& lower & 1.618~\cite{diffserv-np-single,chin-timesharing,hajek-scheduling} & 1.333~\cite{bounded-delay-2-bounded} & 1.25~\cite{chin-timesharing} \\
\bottomrule 
\end{tabular}
}


\subsection{Our Contribution}\label{sec:contrib}
We consider randomized algorithms against an adaptive adversary, motivated
by the following observation.
In reality, traffic through a~switch is not at all independent of the packet
scheduling algorithm.  For example, lost packets are typically resent, and
throughput through a~node affects the choice of routes for data~streams in
a~network. These phenomena~can be captured by the adaptive adversary model but not
by the oblivious one. The adaptive adversary model is also of its own theoretical
interest and has been studied in numerous other settings~\cite{borodin-book}.

The main contribution of this paper is a~simple memoryless scale-invariant
algorithm {\newMIX}, which may be viewed as {\oldMIX}, proposed
by Chin~et~al.~\cite{bounded-delay-chin-journal}, with a~different probability
distribution over pending packets. The competitive ratio of {\newMIX} is at most
$e/(e-1)$ on the one hand, but on the other it is provably better than that for
many restricted variants of the problem.
Some of the upper bounds we provide were known before (cf.~Table~\ref{tbl:results}),
but in general they were achieved by several different algorithms.

Specifically, {\newMIX} is $1/\left(1-(1-\frac{1}{N})^N\right)$-competitive against
adaptive adversary, where $N$ is the maximum, over steps, number of packets that
have positive probability of transmission in the step.
Note that $1/\left(1-(1-\frac{1}{N})^N\right)$ tends to $e/(e-1)$ from below.
The number $N$ can be bounded a~priori in certain restricted variants of the
problem, thus giving better bounds for them, as we discuss in detail in
Section~\ref{sec: implications}. For now let us mention that $N \leq s$ in
$s$-bounded instances and instances with at most $s$ different packet weights.
The particular upper bound of $4/3$ that we obtain for $2$-bounded instances is
tight in the adaptive adversary model~\cite{bounded-delay-2-bounded}.

As is the case with {\oldMIX}, both {\newMIX} and its analysis rely only on the
relative order between the packets' deadlines. Therefore our upper bound(s)
apply to the \emph{Collecting Items} problem~\cite{item-collecting}.
In fact, {\newMIX} is the optimum randomized memoryless algorithm for that
problem in a~strong sense, cf.~Appendix~\ref{sec:lb-in-appx}.



\section{General Upper Bound} 

\subsection{Analysis technique}\label{sec: analysis-intro}

In our analysis, we follow the paradigm of modifying the adversary's buffer, 
introduced by Li et al.~\cite{bounded-delay-li}. Namely, we assume that
in each step the algorithm and the adversary have precisely the same pending
packets in their buffers. Once they both transmit a~packet, we modify the
adversary's buffer judiciously to make it identical with that of the algorithm.
This amortized analysis technique leads to a~streamlined and intuitive proof.

When modifying the buffer, we may have to let the adversary transmit another packet,
inject an extra~packet to his buffer, or upgrade one of the packets in its
buffer by increasing its weight or deadline. We will ensure that these changes will 
be \emph{advantageous to the adversary} in the following sense: for any
adversary strategy $\ADV$, starting with the current step and buffer content, 
there is an adversary strategy $\overline{\ADV}$ that continues computation with
the modified buffer, such that the total gain of $\overline{\ADV}$ from the
current step on (inclusive), on any instance, is at least as large as that of $\ADV$.

To prove $R$-competitiveness, we
show that in each step the expected \emph{amortized gain} of the adversary is at
most $R$ times the expected gain of the algorithm, where the former is the total
weight of the packets that $\ADV$ eventually transmitted in this step.  Both
expected values are taken over random choices of the algorithm.

We are going to assume that {\ADV} never transmits a~packet $a$ if there is
another pending packet $b$ such that transmitting $b$ is always advantageous to
{\ADV}. Formally, we introduce a~dominance relation among the pending packets
and assume that {\ADV} never transmits a~dominated packet.

We say that a~packet $a=(w_a,d_a)$ is \emph{dominated} by a~packet $b=(w_b,d_b)$
at time $t$ if at time $t$ both $a$ and $b$ are pending, 
\(w_a \leq w_b\) and \(d_a \geq d_b\). If one of these inequalities is
strict, we say that $a$ is \emph{strictly dominated} by $b$.
We say that packet $a$ is (strictly) dominated whenever there exists a~packet $b$
that (strictly) dominates it.
Then the following fact can be shown by a~standard exchange argument.
\begin{fact} \label{fact:dominance}
For any adversary strategy $\ADV$, there is a~strategy
$\overline{\ADV}$ with the following properties:
\begin{enumerate}
\item 
the gain of $\overline{\ADV}$ on every sequence is at least the gain of {\ADV}, 
\item 
in every step $t$, $\overline{\ADV}$ does not transmit a~strictly
dominated packet at time~$t$.
\end{enumerate}
\end{fact}
\begin{proof}
$\ADV$ can be transformed into $\overline{\ADV}$ iteratively: take the minimum
$t_0$ such that $\ADV$ first violates the second property in step~$t_0$, and
transform {\ADV} into an algorithm $\ADV'$ with gain no smaller than that of
$\ADV$, which satisfies the second property up to step $t_0$, possibly violating
it in further steps.

Let $t_0$ be the first step in which the second property is violated. Let $y=(w,d)$ 
be the packet transmitted by $\ADV$ and $x=(w',d')$ be the packet that dominates $y$;
then $w' \geq w$ and $d' \leq d$. Let $\ADV'$ transmit the same packets as $\ADV$
up to step~$t_0-1$, but in step $t_0$ let it transmit $x$, and in the remaining
steps let it try to transmit the same packets as $\ADV$. It is impossible in one
case only: when $\ADV$ transmits $x$ in some step $t$. But then $d \geq d' > t$,
so let $\ADV'$ transmit $y$, still pending at $t$. Clearly, the gain of $\ADV'$
is at least as large as the gain of $\ADV$.
\end{proof}
Let us stress that Fact~\ref{fact:dominance} holds for the adaptive adversary
model. Now we give an example of another simplifying assumption, often assumed
in the oblivious adversary model, which seems to break down in the adaptive
adversary model.
In the oblivious adversary model the instance is fixed in advance by the
adversary, so {\ADV} may precompute the optimum schedule to the instance and
follow it. Moreover, by standard exchange argument for the
\emph{fixed} set of packets to be transmitted, {\ADV} may always send the packet
with the smallest deadline from that set---this is usually called the
\emph{earliest deadline first} (EDF) property or order.
This assumption not only simplifies analyses of algorithms but is often crucial
for them to yields desired bounds~\cite{bounded-delay-chin-journal,%
bounded-delay-chrobak-improved,bounded-delay-li,agreeable-moje}.

In the adaptive adversary model, however, the following phenomenon occurs: 
as the instance $I$ is randomized, {\ADV} does not know for sure which packets
it will transmit in the future. Consequently, deprived of that knowledge,
it cannot ensure any specific order of packet transmissions.


\subsection{The Algorithm}

We describe the algorithm's behavior in a~single step.
\begin{algorithm}[H]
\caption{{\newMIX} (single step)}
\begin{algorithmic}[1]
\If{there are no pending packets}
	\State do nothing and proceed to the next step
\EndIf
\State $m \gets 0$ \Comment counts packets that are not strictly dominated
\State $n \gets 0$ \Comment counts packets with positive probability assigned
\State $r \gets 1$ \Comment unassigned probability
\State $H_0 \gets$ pending packets
\State $h_0=(w_0,d_0) \gets$ heaviest packet from $H_0$%
%
\While{$H_m \neq \emptyset$}
	\State $m \gets m+1$
	\State $h_m = (w_m,d_m) \gets$ heaviest not strictly dominated packet from $H_{m-1}$
	\State $p_{m-1} \gets \min\{ 1- \frac{w_m}{w_{m-1}}, \ r \}$
	\State $r \gets r - p_{m-1}$
	\If{$r>0$}
		\State $n \gets n+1$
	\EndIf
	\State $H_m \gets \{ x \in H_{m-1}\ |\ x \text{ is not dominated by } h_m \}$
\EndWhile
\State $p_m \gets r$
%
\State \kw{transmit} $h$ chosen from $h_1,\ldots,h_n$  with probability
distribution $p_1,\ldots,p_n$
\State proceed to the next step
\end{algorithmic}
\end{algorithm}

We introduce the packet $h_0$ to shorten {\newMIX}'s pseudocode by making it
possible to set the value of $p_1$ in the first iteration of the loop.
The packet itself is chosen in such a~way that $p_0=0$, to make it clear that it
is not considered for transmission (unless $h_0=h_1$). The while loop itself could
be terminated as soon as $r=0$, because afterwards {\newMIX} does not assign
positive probability to any packet.
However, letting it construct the whole sequence $h_1,h_2,\dots h_m$
such that $H_m = \emptyset$ simplifies our analysis.
Before proceeding with the analysis, we note a~few facts about {\newMIX}.
\begin{fact}\label{fact:chain-dominance}
The sequence of packets $h_0,h_1,\ldots,h_m$ selected by {\newMIX} satisfies
\begin{align*}
w_0=\ &w_1 > w_2 > \dots > w_m \enspace,\\
	& d_1 > d_2 > \dots > d_m \enspace.
\end{align*}
Furthermore, every pending packet is dominated by one of
$h_1,\ldots,h_m$.
\end{fact}
\begin{fact}\label{fact:probabilities}
The numbers $p_1,p_2,\dots,p_m$ form a~probability distribution
such that
\begin{equation}\label{eq:pbbs-ineq}
p_i \leq 1-\frac{w_{i+1}}{w_i} \qquad\text{ for all } i < m \enspace.
\end{equation}
Furthermore, the bound is tight for $i<n$, while $p_i=0$ for $i>n$, i.e.,
%
%
\begin{equation}\label{eq:pbbs-eq}
p_i = 
\begin{cases}
	1-\frac{w_{i+1}}{w_i}, & \qquad \text{ for } i<n \\
	0, & \qquad \text{ for } i>n
\end{cases}
\end{equation}

\end{fact}
\begin{theorem}\label{thm: main}
{\newMIX} is $1/\left(1-(1-\frac{1}{N})^N\right)$-competitive against an adaptive
adversary, where $N$ is the maximum, over steps, number of packets that are
assigned positive probability in a~step.
\end{theorem}
\begin{proof}
For a~given step, we describe the changes to {\ADV}'s scheduling decisions
and modifications to its buffer that make it the same as {\newMIX}'s buffer.
Then, to prove our claim, we will show that
\begin{align}
	\E\left[\GGADV\right]	&\leq w_1 \enspace, \label{eq: adv-gain}\\
	\E\left[\GRMIX\right]	&\geq w_1 \left(1-(1-\frac{1}{n})^n\right) \enspace, \label{eq: alg-gain}
\end{align}
where $n$ is the number of packets assigned positive probability in the step.
The theorem follows by summation over all steps.

Recall that, by Fact~\ref{fact:dominance}, {\ADV} (wlog) sends a~packet that is
not strictly dominated. By Fact~\ref{fact:chain-dominance}, the packets
$h_1,h_2,\dots h_m$ dominate all pending packets, so the one sent by {\ADV},
say $p$ is (wlog) one of $h_1,h_2,\dots h_m$: if $p$ is dominated by $h_i$,
but not strictly dominated, then $p$ has the same weight and deadline as $h_i$.

We begin by describing modifications to $\ADV$'s buffer and estimate 
$\ADV$'s amortized gain. To this end we need to fix the packet sent
by {\newMIX}, so let us assume it is $h_f=(w_f,d_f)$.
Assume that $\ADV$ transmits a~packet $h_z=(w_z,d_z)$.
We will denote the adversary's amortized gain given the latter assumption
by $\GADV{z}$. We consider two cases. 
\begin{description}
\item{Case~1:} $d_f \leq d_z$.
Then $w_f \leq w_z$, since $h_z$ is not dominated.  After both $\ADV$
and $\newMIX$ transmit their packets, we replace $h_f$ in the buffer
of $\ADV$ by a~copy of $h_z$.
This way their buffers remain the same afterwards, and the change is
advantageous to {\ADV}: this is essentially an upgrade of the packet $h_f$ in
its buffer, as both $d_f \leq d_z$ and $w_f \leq w_z$ hold.

\item{Case~2:} $d_f > d_z$.
After both $\ADV$ and $\newMIX$ transmit their packets,
we let $\ADV$ additionally transmit $h_f$, and we inject a~copy of $h_z$ into
its buffer, both of which are clearly advantageous to {\ADV}.
This makes the buffers of {\ADV} and {\newMIX} identical afterwards.
\end{description}

We start by proving~\eqref{eq: adv-gain}, the bound on the adversary's
expected amortized gain. Note that $\ADV$ always gains $w_z$, and if
$d_z < d_f$ ($z>f$), it additionally gains $w_f$. Thus, when {\ADV} transmits
$h_z$, its expected amortized gain is
\begin{equation} \label{eq: adv-gain-fixed}
	\E\left[\GADV{z}\right] = w_z + \sum_{i<z} p_i w_i \enspace. 
\end{equation}
As the adversary's expected amortized gain satisfies
\begin{equation*}
	\E\left[\GGADV\right] \leq \max_{1\leq i \leq m} \left\{ \E\left[\GADV{i}\right] \right\} \enspace,
\end{equation*}
to establish~\eqref{eq: adv-gain}, we will prove that
\begin{equation} \label{eq: adv-gain-max-bound}
	\max_{1\leq i \leq m} \left\{ \E\left[\GADV{i}\right] \right\} \leq \GADV{1} = w_1 \enspace.
\end{equation}
The equality in~\eqref{eq: adv-gain-max-bound} follows trivially
from~\eqref{eq: adv-gain-fixed}. To see that the inequality
in~\eqref{eq: adv-gain-max-bound} holds as well, observe that,
by~\eqref{eq: adv-gain-fixed}, for all $j < m$,
\begin{equation}
	\E\left[\GADV{i}\right] - \E\left[\GADV{i+1}\right] = w_i - w_{i+1} - p_i w_i
	 \geq 0 \enspace, \label{eq: adv-gain-diff}
\end{equation}
where the inequality follows from~\eqref{eq:pbbs-ineq}.

Now we turn to~\eqref{eq: alg-gain}, the bound on the expected gain of
$\newMIX$ in a~single step. Obviously,
\begin{equation} \label{eq: alg-gain-trivial}
	\E\left[\GRMIX\right] = \sum_{i=1}^n p_i w_i	\enspace.
\end{equation}
By~\eqref{eq:pbbs-eq}, $p_i w_i = w_i - w_{i+1}$ for all $i<n$. Also,
$p_n = 1 -\sum_{i<n} p_i$, by Fact~\ref{fact:probabilities}.
Making corresponding substitutions in~\eqref{eq: alg-gain-trivial} yields
\begin{align}
	\E\left[\GRMIX\right] &= \left(\sum_{i=1}^{n-1} \left(w_i - w_{i+1}\right)\right)
		+ \left(1-\sum_{i=1}^{n-1} p_i\right) w_n \nonumber \\
		&= w_1 - w_n \sum_{i=1}^{n-1} p_i	\enspace. \label{eq:alg-gain-pbbs}
\end{align}
As~\eqref{eq:pbbs-eq} implies $w_i = w_{i-1}(1-p_{i-1})$ for all $i\leq n$,
we can express $w_n$ as
\begin{equation} \label{eq: w_n-as-w_1}
	w_n = w_1 \prod_{i=1}^{n-1} (1-p_i) \enspace.
\end{equation}
Substituting~\eqref{eq: w_n-as-w_1} for $w_n$ in~\eqref{eq:alg-gain-pbbs},
we obtain
\begin{equation} \label{eq: alg-gain-w_1-only}
	\E\left[\GRMIX\right] = w_1 \left(1- \prod_{i=1}^{n-1} (1-p_i)
		\sum_{i=1}^{n-1} p_i \right)	\enspace.
\end{equation}
Note that
\begin{equation*} 
	\sum_{i=1}^{n-1} (1-p_i) + \left( \sum_{i=1}^{n-1} p_i \right) = n-1 \enspace,
\end{equation*}
and therefore the inequality between arithmetic and geometric means yields
\begin{equation} \label{eq: pbbs-prod}
	\prod_{i=1}^{n-1} (1-p_i) \sum_{i=1}^{n-1} p_i \leq (1-\frac{1}{n})^n \enspace.
\end{equation}
Plugging~\eqref{eq: pbbs-prod} into~\eqref{eq: alg-gain-w_1-only} yields
\begin{equation*} 
	\E\left[\GRMIX\right] \geq w_1 \left( 1 - (1-\frac{1}{n})^n \right) \enspace,
\end{equation*}
which proves~\eqref{eq: alg-gain}, and together with~\eqref{eq: adv-gain},
the whole theorem.
\end{proof}


\subsection{Rationale behind the probability distribution}\label{sec:dist}

Recall that the upper bound on the competitive ratio of {\newMIX} is
\begin{equation} \label{eq: ratio}
	\frac{\max_{1\leq z \leq m} \{ \E\left[\GADV{z}\right] \}}{\E\left[\GRMIX\right]} \enspace,
\end{equation}
irrespective of the choice of $p_1,\ldots,p_m$.

The particular probability distribution used in {\newMIX} is chosen to
(heuristically) minimize above ratio by maximizing $\E\left[\GRMIX\right]$,
while keeping~\eqref{eq: adv-gain-max-bound} satisfied, i.e., keeping
$\E\left[\GGADV\right]\leq\GADV{1}=w_1$.

The first goal is trivially achieved by setting $p_1 \gets 1$. This however
makes $\E\left[\GADV{z}\right] > w_1$ for all $z>1$. Therefore, some of the
probability mass is transferred to $p_2,p_3,\ldots$ in the following way.
To keep $\E\left[\GRMIX\right]$ as large as possible, $p_2$ is greedily set
to its maximum, if there is any unassigned probability left, $p_3$ is set to its
maximum, and so on. As $\E\left[\GADV{z}\right]$ does not depend on $p_i$ for
$i \geq z$, the values $\E\left[\GADV{z}\right]$ can be equalized with $w_1$
sequentially, with $z$ increasing, until there is no unassigned probability left.
Equalizing $\E\left[\GADV{j}\right]$ with $\E\left[\GADV{j-1}\right]$ consists
in setting $p_{j-1} \gets 1- \frac{w_j}{w_{j-1}}$, as shown
in~\eqref{eq: adv-gain-diff}. The same inequality shows what is intuitively clear:
once there is no probability left and further values $\E\left[\GADV{z}\right]$
cannot be equalized, they are only smaller than $w_1$.

The lower bound for the \emph{Collecting Items} problem~\cite{item-collecting},
presented in Appendix~\ref{sec:lb-in-appx}, proves that this heuristic does
minimize~\eqref{eq: ratio}.


\subsection{Implications for Restricted Variants}
\label{sec: implications}

We have already mentioned that for $s$-bounded instances or those with at most
$s$ different packet weights, $N \leq m \leq s$ in Theorem~\ref{thm: main}, which
trivially follows from Fact~\ref{fact:chain-dominance}. Thus for either kind
of instances {\newMIX} is $1/\left(1-(1-\frac{1}{s})^s\right)$-competitive.
In particular, on $2$-bounded instances {\newMIX} coincides with the previously
known optimal $4/3$-competitive algorithm
\textsc{Rand}~\cite{bounded-delay-2-bounded}
for the adaptive adversary model.

Sometimes it may be possible to give more sophisticated bounds on $N$,
and consequently on the competitive ratio for particular variant of the problem,
as we now explain. The reason for considering only the packets $h_0,h_1,\ldots,h_m$
is clear: by Fact~\ref{fact:dominance} and Fact~\ref{fact:chain-dominance}, {\ADV}
(wlog) transmits one of them. Therefore, {\newMIX} tries to mimic {\ADV}'s
behavior by adopting a~probability distribution over these packets (recall that
in the analysis the packets pending for {\newMIX} and {\ADV} are exactly the same)
that keeps the maximum, over {\ADV}'s choices, expected amortized gain of {\ADV}
and its own expected gain as close as possible (cf.~Section~\ref{sec:dist}).
Now, if for whatever reason we know that {\ADV} is going to transmit a~packet
from some set $S$, then $H_0$ can be initialized to $S$ rather than all pending
packets, and Theorem~\ref{thm: main} will still hold. And as the upper bound
guaranteed by Theorem~\ref{thm: main} depends on $N$, it might improve if the
cardinality of $S$ is small.

While it seems unlikely that bounds for any restricted variant other than
$s$-bounded instances or instances with at most $s$ different packet weights can
be obtained this way, there is one interesting example that shows it is possible.
For similarly ordered instances (aka~instances with agreeable deadlines)
and oblivious adversary one can always find such set $S$ of cardinality at most
$2$~\cite[Lemma~2.7]{agreeable-moje}; while not explicitly stated, this fact was
proved before by Li~et~al.~\cite{bounded-delay-li}.
Roughly, the set $S$ contains the earliest-deadline and the heaviest packet from
any optimal provisional schedule. The latter is the optimal schedule under the
assumption that no further packets are ever injected, and as such can be found
in any step.


\section{Conclusion and Open Problems}

While {\newMIX} is very simple to analyze, it subsumes almost all
previously known randomized algorithms for packet scheduling and provides new
bounds for several restricted variants of the problem.
One notable exception is the optimum algorithm against oblivious adversary
for $2$-bounded instances~\cite{bounded-delay-chin-journal}. This exposes that
the strength of our analysis, i.e., applicability to adaptive adversary model,
is most likely a~weakness at the same time. The strongest lower bounds on
competitive ratio for oblivious and adaptive adversary differ. And as both are
tight for $2$-bounded instances, it seems impossible to obtain an upper bound
smaller than $4/3$ on the competitive ratio of {\newMIX} for any non-trivial
restriction of the problem in the oblivious adversary model.

In both the algorithm and its analysis it is the respective order of packets'
deadlines rather than their exact values that matter.
Therefore, our results are also applicable to the \emph{Collecting Items}
problem~\cite{item-collecting}, briefly described in
Section~\ref{sec: items}. As mentioned in Section~\ref{sec:contrib}, {\newMIX}
is the optimum randomized memoryless algorithm for \emph{Collecting Items},
cf.~Appendix~\ref{sec:lb-in-appx}.

Therefore, to beat either the general bound of $e/(e-1)$, or any of the
$1/\left(1-(1-\frac{1}{s})^s\right)$ bounds for $s$-bounded instances for buffer
management with bounded delay, one either needs to consider algorithms that are
not memoryless scale-invariant, or better utilize the knowledge of exact
deadlines---in the analysis at least, if not in the algorithm itself.

Last but not least, let us remark again that {\newMIX} and its analysis
might automatically provide better bounds for further restricted variants of the
problem, provided that some insight allows to confine the adversary's choice of
packets for transmission in a~step, while knowing which packets are pending for
it---one such example is the algorithm for similarly ordered instances
(aka~instances with agreeable deadlines)~\cite{agreeable-moje}, as we discussed
in Section~\ref{sec: implications}.


\bibliographystyle{abbrv} 
\bibliography{references}


\newpage
\begin{appendix}

\section{Lower Bound for \emph{Collecting Items}}\label{sec:lb-in-appx}

In this section, for completeness, we evoke the lower bound on the
\emph{Collecting Items} problem~\cite{item-collecting}. As the proof is omitted
in the original article due to space constraints, and the original theorem
statement therein is not parametrized by $N$, we restate the theorem.

\begin{theorem}[Theorem 6.3 of \cite{item-collecting}]
For every randomized memoryless algorithm for the \emph{Collecting Items} problem,
there is an adaptive adversary's strategy using at most $N$ different
packet weights such that the algorithm's competitive ratio against the strategy
is at least $1/\left(1-(1-\frac{1}{N})^N\right)$, and at every step the algorithm
has at most $N$ packets in its queue.
\end{theorem}

Below we present the original proof from~\cite{item-collecting}.

\begin{proof}
Fix some online memoryless randomized algorithm $\A$, and consider the following
scheme. Let $a~> 1$ be a~constant, which we specify later,
and let $n=N-1$ At the beginning, the adversary inserts items
$a^0,a^1,\ldots,a^n$ into the queue, in this order.
(To simplify notation, in this proof we identify
items with their weights.) In our construction we maintain the invariant that in each step,
the list of items pending for $\A$ is equal to $a^0,a^1,\ldots,a^n$.
Since $\A$ is memoryless, in each step it uses the same probability distribution
$(q_j)_{j=0}^n$, where $q_j$ is the probability of collecting item~$a^j$.
Moreover, $\sum_{i=0}^n q_i = 1$, as without loss of generality
the algorithm always makes a~move.

We consider $n+1$ strategies for an adversary, numbered $0,1,\ldots,n$.
The $k$-th strategy is as follows: in each step collect $a^k$,
delete items $a^0,a^1,\ldots,a^k$, and then issue new copies of these items.
Additionally, if $\A$ collected $a^j$ for some
$j>k$, then the adversary issues a~new copy of $a^j$ as well.
This way, in each step exactly one copy of each $a^j$ is
pending for $\A$, while the adversary accumulates in its pending set
copies of the items $a^j$, for $j>k$, that were collected by $\A$.

This step is repeated $T \gg n$ times, and after the last step
both the adversary and the algorithm collect all their pending items.
Since $T \gg n$, we only need to focus on the
expected amortized profits (defined below) in a~single step.

We look at the gains of $\A$ and the adversary in a~single step.
If the adversary chooses strategy~$k$, then it gains $a^k$. Additionally,
at the end it collects the item collected by the algorithm if this item is
greater than~$a^k$. Thus, its \emph{amortized expected
gain} in a~single step is $a^k + \sum_{i>k} q_i a^i$. The expected
gain of $\A$ is $\sum_i q_i a^i$.

For any probability distribution $(q_j)_{j=0}^n$ of the algorithm, the adversary chooses
a~strategy $k$ which maximizes the competitive ratio. Thus, the competitive
ratio of $\A$ is is at least
\begin{eqnarray}	
	R &=& \max_k \braced{ \frac{a^k + \sum_{j > k} q_j a^j}{\sum_j q_j a^j}} 
		 \;\geq\;  \sum_k v_k  \frac{a^k + \sum_{j > k} q_j a^j}{\sum_j q_j a^j} \enspace,
			\label{eqn: rand lower bound, ratio R}
\end{eqnarray}
for any coefficients $v_0,\ldots,v_n\ge 0$ such that $\sum_k v_k = 1$.
Let $M = a^{n+1} - n  (a-1)$. For $k=0,1,...,n$, we choose
\begin{eqnarray*}
	v_k &=& \begin{cases}
		\frac{1}{M}  a^{n-k}  ( a-1 ), 	& \textnormal{if $k < n$} 	\enspace , \\
		\frac{1}{M}  \left( a~- n  (a-1), \right) & \textnormal{if $k = n$} 	\enspace . \\
 	\end{cases}
\end{eqnarray*}
The choice of these values may seem somewhat mysterious, but it's in fact
quite simple{\emdash}it is obtained by considering $\A$'s distributions where
$q_j = 1$ for some $j$ (and thus when $\A$ is deterministic), assuming that
the resulting lower bounds on the right-hand side of
\eqref{eqn: rand lower bound, ratio R} are equal, and solving the resulting
system of equations.

For these values of $v_k$ we obtain
\begin{eqnarray*}
M R \sum_{j=0}^n q_j a^j 
	&\geq&  \sum_{k=0}^n M v_k a^k + \sum_{k=0}^n M v_k \sum_{j>k} q_j a^j \\
	 &=&  \sum_{k=0}^{n-1} M v_k a^k + M v_n a^n + \sum_{j=0}^n q_j a^j \sum_{k < j} M  v_k \\
	 &=& n(a-1)a^n + [a~- n(a-1)]a^n
				+ \sum_{j=0}^n q_j(a^j-1)a^{n+1} \\
	 &=& a^{n+1}
				+ a^{n+1} \sum_{j=0}^n q_ja^j - a^{n+1}\sum_{j=0}^n q_j \\
	 &=& a^{n+1}
				+ a^{n+1} \sum_{j=0}^n q_ja^j - a^{n+1}\\
	 &=&  a^{n+1}  \sum_{j=0}^n q_j  a^j \enspace.
\end{eqnarray*}
Therefore, $R \geq a^{n+1} / M$. This bound is maximized for $a~= 1+1/n$, in which case 
we get
\begin{eqnarray*}
R \;\geq\;
\frac{\left(1+\frac{1}{n}\right)^{n+1} }{ \left(1+\frac{1}{n}\right)^{n+1} - 1 } 
 = \frac{\left(1+\frac{1}{N-1}\right)^{N} }{ \left(1+\frac{1}{N-1}\right)^{N} - 1 }
 = \frac{1}{1-\left(1-\frac{1}{N}\right)^N}
\enspace ,
\end{eqnarray*}
completing the proof.
\end{proof}

\end{appendix}


\end{document}